\documentclass[aps,pra,twocolumn,groupedaddress,nofootinbib]{revtex4-1}
\usepackage{graphicx}
\usepackage{epstopdf}

\usepackage{tikz}
\usetikzlibrary {positioning}
\usepackage{pgf}
\usepackage{amsmath,amssymb,amsfonts,amsthm}
\usepackage{mathtools}
\usepackage{multirow}
\usepackage{verbatim}
\usepackage{url}
\usepackage{comment}
\usepackage{mathtools}
\usepackage{epsf,pgf,graphicx}
\usepackage{color}
\usepackage{tikz,pgf}
\usepackage{diagbox}
\usepackage{makecell}
\usepackage{color}
\usepackage{braket}
\usepackage{mathtools}
\usepackage[capitalize]{cleveref}
\usepackage{xcolor}
\newtheorem{theorem}{Theorem}
\newtheorem{lemma}{Lemma}

\newtheorem{corollary}{Corollary}

\usepackage{subfig}
\newtheorem*{definition}{Definition}
\newtheorem{remark}{Remark}

\newcounter{casenum}
\newenvironment{caseof}{\setcounter{casenum}{1}}{\vskip.5\baselineskip}
\newcommand{\scase}[2]{\vskip.5\baselineskip\par\noindent {\bfseries Case \arabic{casenum}:} #1\\#2\addtocounter{casenum}{1}}

\newcommand{\B}{\ensuremath{\mathcal{B}}}
\newcommand{\R}{\ensuremath{\mathbb{R}}}
\newcommand{\op}{\ensuremath{\oplus}}

\begin{document}
\title{Classical-Quantum Separations in Certain Classes of Boolean Functions\\ -- Analysis using the Parity Decision Trees}
\author{Chandra Sekhar Mukherjee}
\email{chandrasekhar.mukherjee07@gmail.com}
\affiliation{ Indian Statistical Institute, Kolkata}
\author{Subhamoy Maitra} 
\email{subho@isical.ac.in }
\affiliation{ Indian Statistical Institute, Kolkata}

\begin{abstract}
In this paper we study the separation between the deterministic (classical) query complexity ($D$)
and the exact quantum query complexity ($Q_E$) of several
Boolean function classes using the parity decision tree method.
We first define the Query Friendly (QF) functions on $n$ variables as the ones with minimum deterministic 
query complexity $(D(f))$. We observe that for each $n$, there exists a non-separable class of QF functions such that 
$D(f)=Q_E(f)$. Further, we show that for some values of $n$, all the QF functions are non-separable. Then we present QF 
functions for certain other values of $n$ where separation can be demonstrated, in particular, $Q_E(f)=D(f)-1$. In 
a related effort, we also study the Maiorana McFarland (M-M) type Bent functions.
We show that while for any M-M Bent function $f$ on $n$ variables $D(f) = n$, separation can be achieved as 
$\frac{n}{2} \leq Q_E(f) \leq \lceil \frac{3n}{4} \rceil$. Our results highlight how different classes of Boolean 
functions can be analyzed for classical-quantum separation exploiting the parity decision tree method.
\end{abstract}
\maketitle
\section{Introduction}
Query Complexity is a model of computation in which a function 
$f(x_1,x_2,\ldots ,x_n):\{0,1\}^n \rightarrow \{0, 1\}$ is evaluated 
using queries to the variables $x_i, \ 1 \leq i \leq n$. The query complexity model has been 
widely studied under different computational 
scenarios, such as classical deterministic model and exact quantum model \cite{Bool1}. 
While the study can be conducted for functions with any finite range, 
Boolean functions are most widely studied in this area, for their simplicity as well as 
the richness in terms of generalization. 
Substantial work has been completed on asymptotic separation of query complexity under different 
models~\cite{AMB1,ANDk,AMB2} and in finding separation between 
classical deterministic and exact quantum query complexity models 
for different Boolean functions, such as the  ${\sf EXACT^n_{k,l}}$~\cite{AMB4} functions.
One may note that the query complexity of a Boolean function does not necessarily relate to its optimal circuit depth.
However, in many cases the circuit obtained using the query complexity model remains the most optimal till date.
The query complexities of functions under different computational models also form a better picture of the advantage
offered by quantum computers in function evaluation.  
In this regard finding classical-quantum separation and the query complexity of Boolean functions
in different computational model remains a non trivial and interesting problem. 
We analyze this separation for two Boolean Function classes.
The first class comprises of Query Friendly functions, which we define for a given $n$
as Boolean functions with $n$ influencing variables with least possible deterministic
query complexity. 
The second class is that of the Maiorana McFarland (M-M) Bent functions~\cite{MM1} on $n$
variables, which is a 
large class of cryptographically important Boolean functions. 
The study of (M-M) Bent functions as a generalized class is also interesting because
this class consists of many functions which are not isomorphic to each other.

A common method of forming Quantum Algorithms in the query complexity model
is using the parity function to calculate the parity of two bits $x_{i_1} \op x_{i_2}$ using a single query.
This method is also practical for implementation in a noisy quantum computer as the bits
in superposition are measured after each oracle access. 
In this direction we use this method of calculating parity of variables in a disciplined manner along with 
combinatorial reductions to find separation in query complexity between the classical deterministic 
and the exact quantum model in the aforementioned Boolean Function classes.

We now introduce some notations and the concept of classical and quantum oracles and 
describe the deterministic classical and exact quantum query model in details.

{\bf Algebraic Normal Form (ANF):} It is known that given any total Boolean function, there exists a unique
multivariate polynomial defined over GF(2) which exactly defines the function.
Formally, one can write, $$f(x_1,x_2,\ldots,x_n)=\bigoplus_{\mathbf{a} = (a_1,\ldots,a_n) \in \{0, 1\}^n}\lambda_{\mathbf{a}}(\prod_{i=1}^n x_i^{a_i}),$$
where $\lambda_{\mathbf{a}}\in \{0, 1\}$ and $x_1, \ldots, x_n \in \{0, 1\}$. 
The Hamming weight of $\mathbf{x} \in \{0, 1\}^n$, $wt(\mathbf{x})$, is defined as $wt(\mathbf{x}) = \sum_{i=1}^n x_i$ where the sum is 
over ring of integers. The algebraic degree of $f$, $\deg(f)$, is defined as 
$\deg(f)=\max_{\mathbf{a} \in \{0, 1\}^n}\{wt(\mathbf{a}): \lambda_{\mathbf{a}}\neq 0\}$.

We also define the term influencing variables in this context. We call a variable $x_i$ of a function $f(x_1,x_2,\ldots, x_n)$ influencing
if there exists a set of values $\{x_1, x_2, \ldots, x_{i-1}, x_{i+1}, \ldots, x_n\}$ such that
\begin{align*}
&f(x_1, x_2, \ldots, x_{i-1}, 0, {x_{i+1}}, \ldots, x_n) \neq \\
&f(x_1, x_2, \ldots, x_{i-1}, 1, {x_{i+1}}, \ldots, x_n).
\end{align*}

The number of influencing variables is also represented as the number
of variables present in the ANF of the corresponding function.

It is also important to note that the algebraic degree of a Boolean function is different 
from the polynomial degree of the function, which commonly used in obtaining lower bounds 
of query algorithms.

\ \\
{\bf Classical and Quantum Oracle:}
In the query complexity model, the value of any variable can only be queried using an oracle.
An oracle is a black-box which can perform a particular computation.
In the classical model, an oracle accepts an input $i \ (1 \leq i \leq n)$ and output the value of the variable $x_i$. 
In the quantum model, the oracle needs to be reversible. It is represented as an unitary $O_x$ which functions as follows.
\begin{align*}
O_x\ket{i}\ket{\phi}=\ket{i}\ket{\phi \oplus x_i},~1 \leq i \leq n \\
\end{align*}
Figure \ref{fig:qc0} represents the working of an oracle in the quantum complexity model, which is similar to what is presented 
in~\cite[Fig. 3]{oneQ}

\begin{figure}[!htb]
\begin{center}
\includegraphics[scale=0.4]{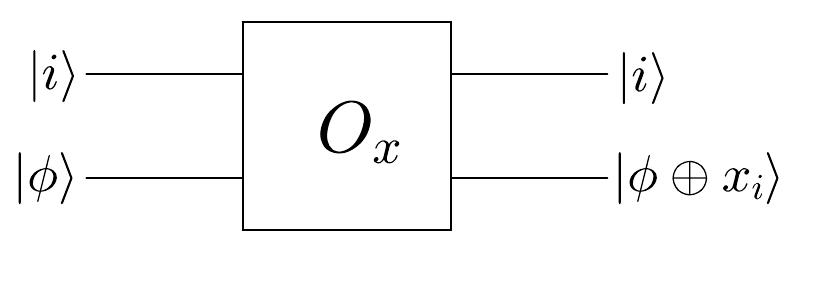}
\end{center}
\caption{Working of a quantum oracle}
\label{fig:qc0}
\end{figure}

The query complexity of a function is the maximum number of times this oracle needs to be used to evaluate the value of the function $f$ 
for any value of the variables $x_1, x_2, \ldots, x_n$. We will be focusing on total Boolean functions from here on, i.e., 
$f : \{0,1\}^n \rightarrow \{0,1\}$. Let us now specify the models.

\ \\
{\bf Deterministic (Classical) Query Complexity:}
The minimum number of queries that a function $f$ needs to be evaluated using a deterministic algorithm
is called its Deterministic Query Complexity ($D(f)$). We generally omit the word `classical'. 
A query based classical deterministic algorithm for evaluating a Boolean function $f:\{0,1\}^n \rightarrow \{0,1\}$ 
can be expressed as a rooted decision tree as follows. 

In this model, every internal node corresponds to a query to a variable $x_i \ 1 \leq i \leq n$.
Each leaf is labeled as either $0$ or $1$.
The tree is traversed from the root of the tree till it reaches a leaf in the following manner.
Every internal node has exactly two children and depending on 
the outcome of the query ($0$ or $1$ respectively), one of the two children are visited (left or right, respectively). 
That is this is a binary tree. The leaf nodes correspond to the output of $f$ for different inputs. 
Every decision tree uniquely defines a Boolean function which we can obtain by deriving the 
Algebraic Normal Form (ANF) from a given tree. 
For example, the ANF of the Boolean function corresponding to the tree shown in Figure \ref{fig:tree1}
is $(x_1 \op 1)(x_2) \op x_1(x_3 \op 1)=x_1x_2 \op x_1x_3 \op x_1 \op x_2 \op x_3$.

\begin{figure}[!htb]
\begin{center}
\includegraphics[scale=0.4]{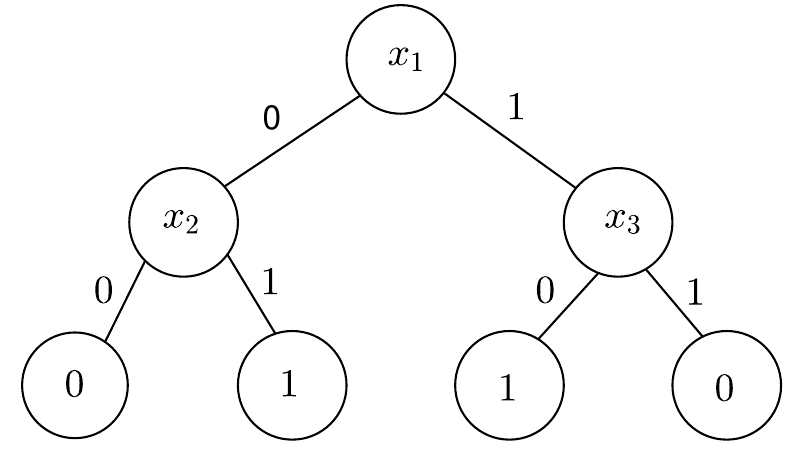}
\end{center}
\caption{Example of a decision tree}
\label{fig:tree1}
\end{figure}

Corresponding to a function, there can be many Deterministic Query Algorithms that can evaluate it.
The depth of a decision tree is defined as the number of edges 
encountered in the  longest root to leaf path. 
Given $f$, the shortest depth decision tree representing the function, is called the optimal decision 
tree of $f$ and the corresponding depth is termed as the Deterministic classical complexity of $f$, denoted as
$D(f)$. We further describe the following notations related to the decision tree model.
\begin{itemize}
\item Let there be a decision tree corresponding to a function $f$
such that no variable appears twice in the tree.
We can then identify an internal node with the variable that it queries.
In such a case we use the notation $val(x_i,c)$ to denote
the left or right children of the internal node which queries the variable $x_i$
where c is $0$ or $1$, respectively.

\item We also define a fully-complete binary tree.
We call a binary tree fully-complete if it has depth $k$ and there are total 
$2^k-1$ internal nodes, i.e. a fully complete binary tree is a $k$-depth decision tree in which
every internal node has two children and all the nodes in the $k$-th level
are parents of leaves.
It is to be noted that this differs from a complete binary tree,
in which every level other than the last is completely filled,
and in the last level all the nodes are as far left as possible.
\end{itemize}

\ \\
{\bf Exact Quantum Query Complexity :}
A Quantum Query Algorithm is defined using a start state $\ket{\psi_{start}}$
and a series of unitary Transformations $$U_0, O_x, U_1, O_x, \ldots, U_{t-1}, O_x, U_t,$$ 
where the unitary operations $U_j$ are indifferent of the values of the variables $x_i$ 
and $O_x$ is the oracle as defined above.
Therefore, the final state of the algorithm is $$\ket{\psi_{final}}= U_t O_x U_{t-1} \ldots U_1 O_x U_0 \ket{\psi_{start}}$$
and the output is decided by some measurement of the state  $\ket{\psi_{final}}$.
A quantum algorithm is said to exactly compute $f$ if for all $(x_1, x_2, \ldots, x_n)$ it outputs the value of the function
correctly with probability 1.
The minimum number of queries needed by a Quantum Algorithm to achieve this is 
called the Exact Quantum Query Complexity $Q_E(f)$ of the function.

\ \\
{\bf Isomorphism( PNP equivalence):} Two functions $f$ and $g$ over $\{0, 1\}^n$
are called isomorphic (PNP equivalent) if the ANF of $f$ can be derived from ANF of $g$
by negation and permutation of the input variables of $g$
and by adding the constant term $1$ in the ANF, that is negation of the output.
If $f$ and $g$ are isomorphic then $D(f)=D(g)$ and $Q_E(f)=Q_E(g)$~\cite[Section 2.2]{exact}.

\ \\
{\bf Separability :}
A Boolean function $f$ is called separable if $Q_E(f)<D(f)$ and non-separable otherwise.

In this paper we concentrate on the deterministic and exact quantum query complexity of different 
Boolean function classes. 
There are other computational models such as the classical randomized model 
and the bounded error quantum model \cite{AMB3} and there exists rich literature on work on these models as well. 
However, those are not in the scope of this work. 

In this regard one may note that the work by Barnum et.al~\cite{sdp} can be used 
to find the exact quantum query complexity of any function on $n$ variables by 
repetitively solving semi definite programs (SDP). Montanaro et.al~\cite{exact} have used this 
method to find exact quantum query complexity of all Boolean functions 
upto four variables as well as describe a procedure of formulating the quantum algorithm 
to achieve the said exact quantum query complexity. 
This method is not yet found to be suitable for finding 
the exact quantum query complexity of a general classes of Boolean functions.
Additionally, the SDP are resource intensive in nature and solving the SDP for large values of $n$
is computationally challenging. 
But for the cases where the number of variables is low, this does offer an exhaustive view of the
exact quantum query complexities of all Boolean functions.

As an example, in a very recent paper Chen et.al~\cite{oneQ} have shown that $f(x)=x_i$ or $f(x)=x_{i_1} \op x_{i_2}$
are the only Boolean functions with $Q_E(f)=1$.
However the work of Montanaro et.al~\cite[Section 6.1]{exact} show that
the Boolean functions $f$ with $2$ or lesser variables and $Q_E(f)=1$ are
\begin{itemize}
\item The single variable function $x_i$.
\item The two variable functions $x_{i_1} \op x_{i_2}$.
\end{itemize}
Then it is shown in~\cite[Section 6.2]{exact} that the minimum quantum exact 
quantum query complexity of any Boolean function with $3$ or more influencing variables is $2$.
This essentially implies that the work of ~\cite{oneQ} is in fact a direct corollary of~\cite{exact}.

\subsection{Organization \& Contribution}
In Section~\ref{sec:2}, we start by describing the fact that the maximum number
of influencing variables that a function with $k$ deterministic query complexity
can have is $(2^k-1)$. We first construct such a function using the decision tree model. 
The decision tree representation of such a function is a $k$-depth fully-complete 
binary tree in which every internal node queries a unique variable. 
We first prove in Theorem~\ref{th:2} that any function with $2^k-1$ influencing variables 
and $k$ deterministic query complexity must have the same exact quantum query complexity ($k$). 

Next, we define a special class of Boolean functions in Section~\ref{subsec:0}, called the ``Query Friendly" functions. 
A function $f$ with $n$ influencing variables is called query friendly if there does not exist 
any other function with $n$ influencing variables with lesser deterministic query complexity than $f$.
If $n$ lies between $2^{k-1}$ and $2^k-1$ (both inclusive) then all functions with 
deterministic query complexity $k$ are called query friendly functions. 
The proof in Theorem~\ref{th:2} directly implies that all query friendly functions 
with $n=2^k-1$ influencing variables are non-separable.

Then in Section~\ref{subsec:1} we identify a class of non-separable query friendly
functions for all values of $n$. We conclude this section by showing 
that all query friendly functions with $n=2^k-2 ~(k>2)$ influencing variables are non-separable as well.

In Section \ref{sec:3}, we describe the parity decision tree model.
We first discuss the simple result that a $k$-depth parity decision tree
can describe functions with upto $2^{k+1}-2$ influencing variables.
In Section~\ref{subsec:2} we define another set of query friendly functions on $n$ influencing variables 
that exhibit minimum separation (i.e., one)
between deterministic and exact quantum query complexity for certain generalized values of $n$.
We prove by construction that if $2^{k-1} \leq n <2^{k-1}+2^{k-2}$ then there exists a class of query friendly
functions such that for any function $f$ in that class we have $Q_E(f)=D(f)-1$.
One should observe that although we prove this separation for a particular function for any $n$, this 
implicitly proves separation for a class of Boolean functions, as reemphasized in Remark~\ref{remark1}.
We conclude the section by showing that for other values of $n$ there does not exist separable 
query friendly functions that can be completely described by the parity decision tree model.

Next in Section~\ref{sec:4} we study the Maiorana McFarland (M-M) type Bent functions, 
which is a cryptographically important class of Boolean Functions. 
This class is interesting as the algebraic degree of functions of this class
defined on $n$ variables vary between $2$ and $\frac{n}{2}$. 
First we observe that the deterministic query complexity of any function of this type on
$n$ variable is $n$. We further observe that the parity decision tree method can be used to 
form a simple algorithm that needs $\lceil \frac{3n}{4} \rceil$ queries for any function in this class.
We conclude this section by describing the real polynomial that describes any function 
belonging to this class, which gives us a lower bound of $\frac{n}{2}$ for the 
exact quantum query complexity of any function belonging to this class.

We conclude the paper in Section~\ref{sec:5} outlining the future direction of our work.
We further state open problems that we have encountered in this work. Solution
to these problems will help us understand the limitations of the parity decision tree model
as well as get possibly more optimal quantum algorithms for different classes of M-M functions.  

\section{Decision Trees and No-separation results}\label{sec:2}
As we have discussed, query algorithms can be expressed as decision trees in the classical deterministic model.
In this regard, let us present the two following simple technical results.
These results are well known in folklore and we present them for completeness.
\begin{lemma}\label{lemma:2}
There exists a Boolean function $f_k$ with $2^k-1$ influencing variables such that $D(f) \leq k$.
\end{lemma}
\begin{proof}
We construct this function for any $k$ as follows. We know that if a Boolean function $f$ can be 
expressed as a decision tree of depth $d$, then $D(f) \leq d$.
We now build a decision tree, which is a fully-complete binary tree of depth $k$. 
Each of the internal nodes in this tree is a unique variable,
that is, no variable appears in the decision tree more than once. 
Since there are $2^k-1$ internal nodes in such a tree, this decision tree represents 
a Boolean function $f_k$ on $2^k-1$ variables with $D(f_k) \leq k$.
 
Without loss of generality we can name the root variable
of the corresponding decision tree as $x_1$ 
and label the variables from left to right at each level in ascending order.
The resultant structure of the tree is shown in Figure \ref{fig:maxtreed}.
\end{proof}

\begin{figure}[!htb]
\begin{center}
\includegraphics[scale=0.28]{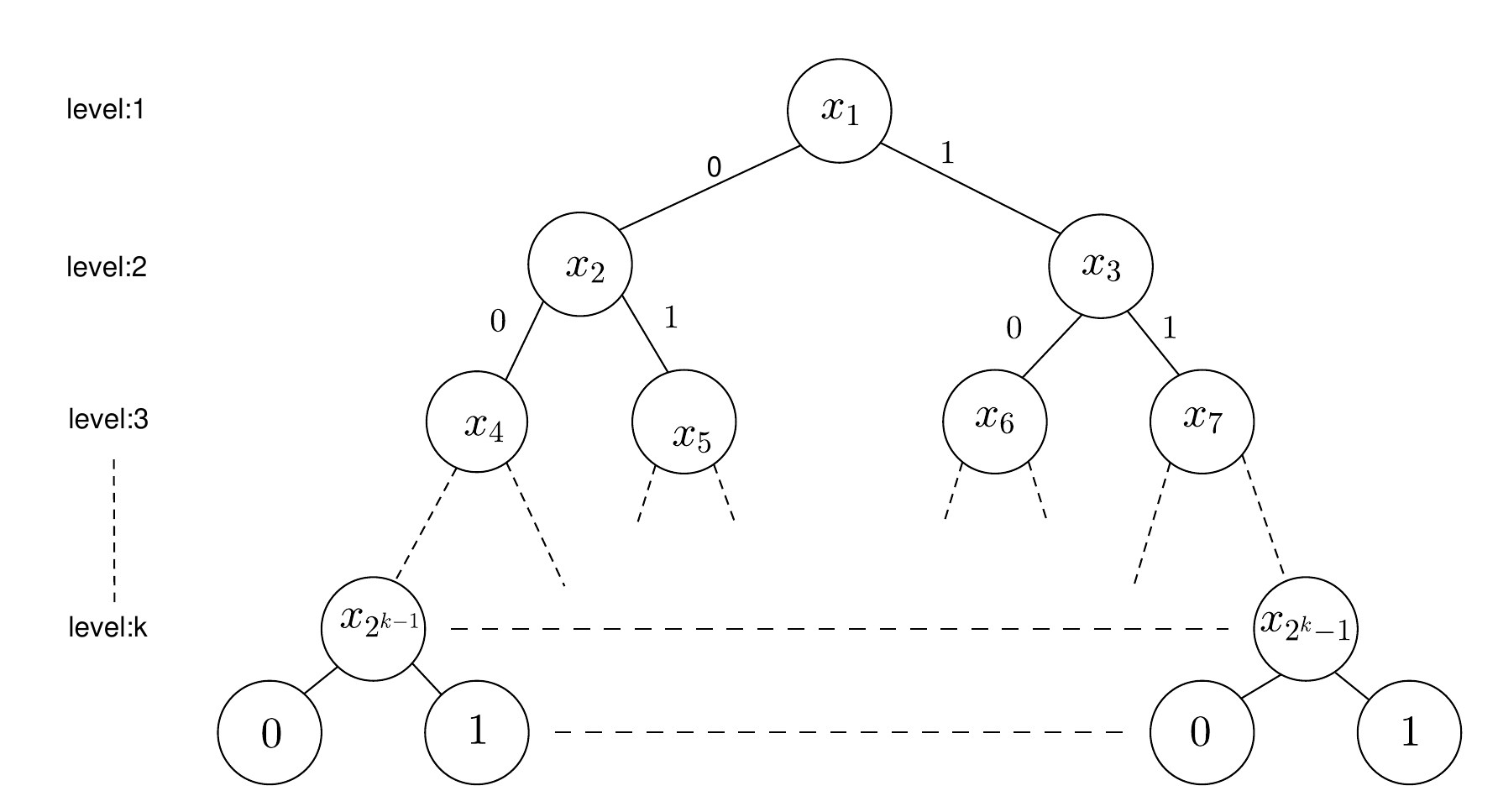}
\end{center}
\caption{Decision Tree corresponding to function $f$ with maximum influencing variables for $D(f)=k$}
\label{fig:maxtreed}
\end{figure}

Having constructed such a Boolean function $f_k$,
we now show that is indeed the function with the maximum number of influencing variables
that can be evaluated using the deterministic computational model using $k$ queries.

\begin{lemma}\label{lem:l0}
Given any integer $k$, the maximum number of influencing variables that a Boolean function $f$ has 
such that $D(f)=k$ is $2^k-1$.
\end{lemma}
\begin{proof}
Suppose there exists a Boolean function with $n_1 ( > 2^k-1)$ influencing variables
that can be evaluated using $k$ queries. This implies that there exists a corresponding 
decision tree of depth $k$ that expresses this function.
However, in a decision tree corresponding to a Boolean function $f$, all the influencing variables 
should be present as an internal node at least once in the decision tree.
Otherwise,  
\begin{align*}
&f(x_1, x_2, \ldots, x_{i-1}, 0, \ldots, {x_n}) \\
&=f(x_1, x_2, \ldots, x_{i-1}, 1, \ldots, {x_n})~ \forall x_j \in \{0,1\}: j \neq i,
\end{align*}
which implies that $x_i$ is not an influencing variable of the function.
Since there cannot exist a decision tree of depth $k$ that has more than $2^k-1$ internal nodes, such a function can not exist.

This implies that for any function $f$ with $n=2^k-1$ influencing variables and $D(f)=k$,
the corresponding decision tree is a $k$-depth complete tree where every variable
is queried only once.
\end{proof}
It immediately follows that a function $f$ with $n=2^k-1$ influencing variables has
deterministic query complexity $D(f) \geq k$. 

\begin{theorem}\label{th:2}
Given any Boolean function $f$ with $2^k-1$ influencing variables and $D(f)=k$ we have $Q_E(f) = k$.
\end{theorem}
\begin{proof}
This is proven by showing that any function $f$ characterized as above 
is at least as hard to evaluate as the function $AND_k$, 
which is AND of $k$ variables.

Given such a function $f$, there exists a corresponding $k$-depth complete tree $T_f$.
As we have shown in Lemma~\ref{lem:l0}, in such a tree all internal nodes will query a variable
and all the variables will appear in the tree exactly once.

Given the decision tree $T_f$ corresponding to $f$ let 
$x_{i_1}, x_{i_2}, x_{i_3}, \ldots, x_{i_k}$ be a root to internal node path in the tree
so that children of $x_{i_k}$ are the leaf nodes. Here
\begin{align*}
val(x_{i_t},1)=x_{i_{t+1}}, \ 1 \leq t \leq k-1.
\end{align*} 
We call this set of variables $s_{max}$. 
We fix the values of the variables
$\{x_1, x_2, x_3, \ldots, x_{2^k-1} \} \setminus s_{max}$ as follows.
Each of the variables at a level less than or equal to $k-1$ is assigned either $0$ or $1$.
Now either 
$val(x_{i_k},0)=0 \text{ and } val(x_{i_k},0)=1$
or
$val(x_{i_k},0)=1 \text{ and } val(x_{i_k},0)=0$.

\begin{itemize}
\item In the first case, If a variable is at the $k$-th level, 
i.e., its children are the leaf nodes then each such variable $y_i$
is fixed at the value $c_i$ such that $val(y_i,c_i)=0$.
Then the function is reduced to $\displaystyle \prod_{t=1}^k x_{i_k}$.

\item In the second case, the values of variables $y_i$
in the $k$-th level is fixed at the value $e_i$ so that 
$val(y_i,c_i)=1$.
Then the function is reduced to  $\Big( \displaystyle \prod_{t=1}^k x_{i_k} \Big) \op 1$. 
\end{itemize} 
The reduced function is $AND_k$ in the first case and $OR_k$ in the second case.
In both the cases we have $Q_E(f) \geq k$, as $Q_E(AND_k)= Q_E(OR_k)= k$~\cite[Table 1]{ANDk}.
We also know that $Q_E(f) \leq D(f)$ for any Boolean function $f$ and therefore $Q_E(f) \leq k$.
Combining the two we get $Q_E(f)=k$.
\end{proof}

We reiterate the idea behind the proof to further simplify the argument.
Reducing a function to $AND_k$ essentially implies that there exists 
a set of variables $x_1,x_2,\ldots x_k$, such that if they are not all equal to $1$,
then the function outputs $0$. 
In terms of the tree the implication is as follows.
Let the path in the proof of Theorem~\ref{th:2} be $x_{i_1},x_{i_2}, \ldots x_{i_lk}$
such that the function is reduced to $AND_k$ by fixing values of the other variables.
Then while the decision tree is traversed from the root, if any of these $k$ variable's
value is $0$, we move to a node that is out of the path, and then the value of the 
other internal nodes should be so fixed that we always reach a $0$-valued leaf node.

\subsection{\label{subsec:0}Query Friendly Functions}
Having established these results, we characterize a special class of Boolean functions. 
Given any $n$, We call the Boolean functions with $n$ influencing variables that have 
minimum deterministic query complexity as the query friendly functions on $n$ variables. 
We denote the corresponding query complexity of this class of functions as 
$DQ_n$, and its value is calculated as follows.
\begin{lemma}
The value of $DQ_n$ is equal to $\lceil \log (n+1) \rceil$.
\end{lemma}
\begin{proof}
We consider any $n$ such that $2^{k-1}-1 < n \leq 2^k-1$.
We have shown in lemma~\ref{lem:l0} that there cannot exist a Boolean function with $n$ variables
that can be evaluated with $k-1$ classical queries.

Since the maximum number of influencing variables that a Boolean function with $k$ query complexity 
has is $2^k-1$ as proven above, there exists a Boolean function with $n$ variables with $D(f)=k$.
Now $\lceil \log (n+1) \rceil=k$, which concludes the proof.
\end{proof}

\begin{corollary}
For $n=2^k-1$, there does not exist any separable query friendly functions.
\end{corollary}
\begin{proof}
For $n=2^k-1$, we have $DQ_n=k$.
We have shown in Theorem ~\ref{th:2}  that any function $f$ with $2^k-1$ influencing variables
and $D(f)=k$ has $Q_E(f)=k$.
\end{proof}

Now let us provide some examples of such functions where the deterministic classical and exact quantum query complexities
are equal.
\begin{itemize}
\item $k=2, n= 2^k - 1 = 3, Q_E(f) = D(f) = 2$: the function is $f = (x_1 \op 1) x_2 \op x_1 x_3= x_1x_2 \op x_1x_3 \op x_2$.

\item $k=3, n= 2^k - 1 = 7, Q_E(f) = D(f) = 3$: the function is $f=(x_1\op 1)((x_2\op 1) x_4 \op x_2 x_5) \op x_1 ((x_3 \op 1) x_6 \op x_3 x_7) = x_1x_2x_4 \op x_1x_2x_5 \op x_1x_4 \op x_2x_4 \op x_2x_5 \op x_4 \op x_1x_3x_6 \op x_1x_3x_7 \op x_1x_6$.
\end{itemize}
Next we move to a generalization when $n \neq 2^k-1$.

\subsection{Extending the result for $n \neq 2^k - 1$}
\label{subsec:1}
We first identify a generic set of non-separable query friendly functions where $2^{k-1}-1 < n < 2^k-1$
and then show that no query friendly function on $n=2^k-2, k>2$~influencing variables are separable.
We define such a set of non-separable query friendly functions for
$2^{k-1}-1 < n < 2^k-1$ using the decision tree model again.
We construct a decision tree of depth $k$ such that the 
first $k-1$ levels are completely filled and every variable occurs exactly once in the decision tree.
That implies there are $n-2^{k-1}+1$ nodes in the $k$-th level.
Let us denote the corresponding function as $f_{(n,1)}$. 
\begin{theorem}
The Boolean function $f_{(n,1)}$ on $n$ influencing variables has $D(f_{(n,1)})=Q_E(f_{(n,1)})$.
\end{theorem}
\begin{proof}
This $k$-depth decision tree constructed for any $n$ such that $2^{k-1}-1 < n <2^k-1$ has the following properties.
\begin{itemize}
\item The corresponding function has deterministic query complexity equal to $k$. This is because the number 
of influencing variables in the function is more than the number of variables that a Boolean function 
with deterministic query complexity $k-1$ can have.

\item There is at least one internal node at $k$-th level. let that node be called $x_{i_k}$.
Let the root to $x_{i_k}$ path be $x_1,x_{i_2},x_{i_3}, \ldots, x_{i_{k-1}},x_{i_k}$
such that $val(x_1,d_1)=x_{i_2},val(x_{i_2},d_2)=x_{i_3}$ and so on. 
Applying the reduction used in Theorem~ \ref{th:2} the corresponding Boolean function can be reduced to
the function $(x_1 \op \overline{d_1})(x_{i_2} \op \overline{d_2})\ldots (x_{i_k} \op \overline{d_k})$ 
which is isomorphic to $AND_k$. 
(Note that $\overline{d_i} = 1 \op d_i$, i.e., the complement of $d_i$.)
This implies that $Q_E(f_{n,1}) \geq k$. We also know $D(f_{n,1})=k$,
and therefore the exact quantum query complexity of the function is $k$. 
Figure \ref{fig:6} gives an example of a function in $f_{(5,1)}$.
\end{itemize}
\end{proof}
\begin{figure}
\begin{center}
\includegraphics[scale=0.20]{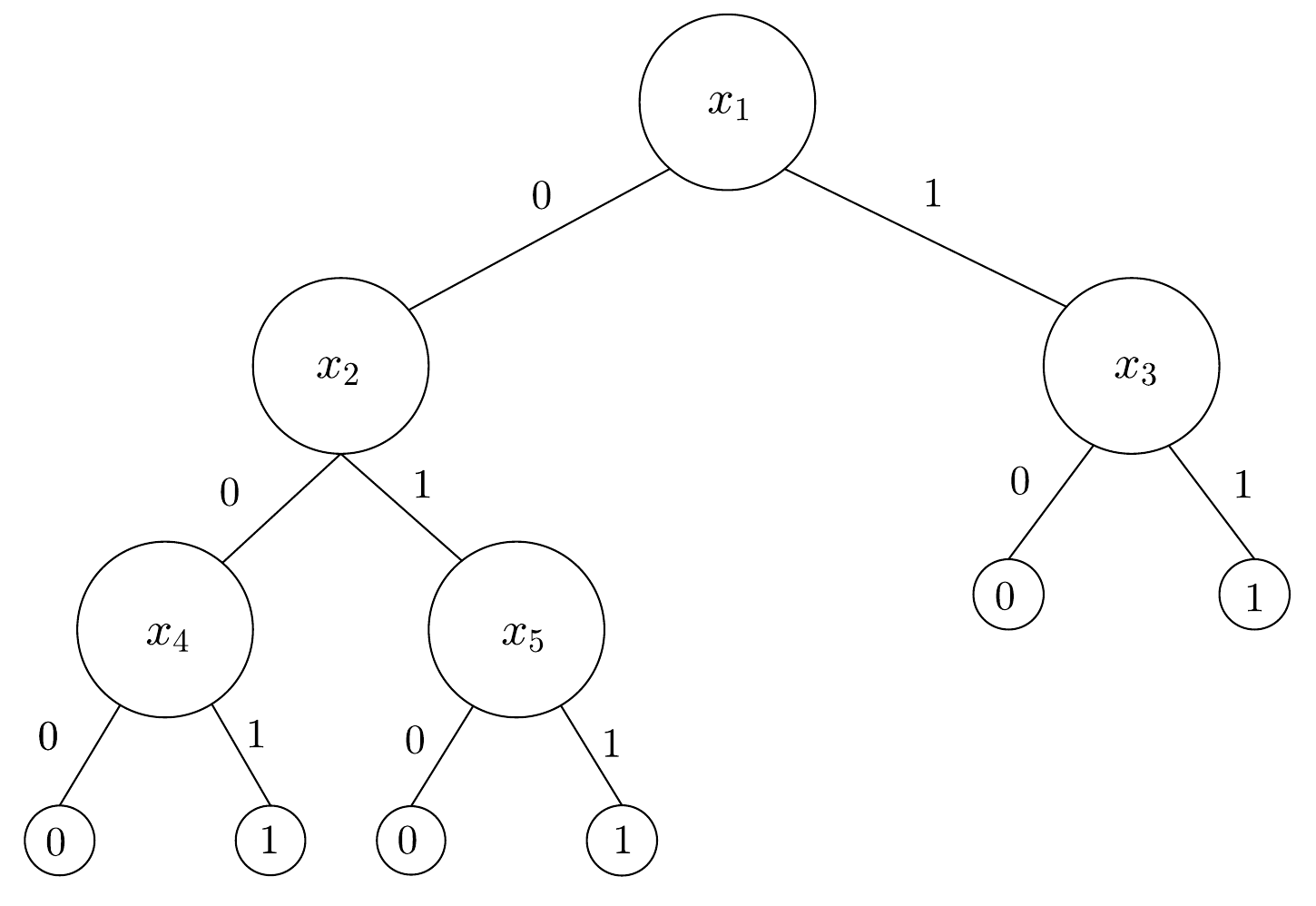}
\end{center}
\caption{Decision Tree corresponding to $f_{(5,1)}$}
\label{fig:6}
\end{figure}
The result is thus a generalization when $2^{k-1}-1<n<2^k-1$, in identifying a class of functions where the separation 
between classical and quantum domain is not possible.
 
We now show that in fact for $k>2$, all query friendly functions with $2^k-2$ variables are non-separable.
\begin{theorem}
Let $f$ be a query friendly function on $n=2^k-2$ variables, such that $k>2$. Then $D(f)=Q_E(f)$.
\end{theorem}
\begin{proof}
There exists a decision tree $T_f$ of depth-$k$ that evaluates $f$.
Since $f$ has $2^k-2$ variables then $T_f$ can be of the following forms:
\begin{enumerate}
\item $T_f$ has $2^k-2$ internal nodes and each of the nodes query a unique variable. Each tree of this type
corresponds to a function of the type $f_{(n,1)}$ and therefore is non-separable. 

\item $T_f$ has $2^k-1$ internal nodes and there exists two nodes in the tree which query the same variable.
\end{enumerate}
We analyze the different structures of $T_f$ corresponding to the second case.
Let the root node queries a variable $x_1$.
Then the following cases can occur.

\begin{caseof}

\scase{Both the children of $x_1$ query the same variable $x_2$.}
{
Let the two nodes be represented by $x_2^0$ and $x_2^1$
We choose a $k$-depth path $x_1,x_2^1,x_3,\ldots x_k$ such that
$val(x_i,1)=x_{i+1}, 1 \leq i \leq k-1$.
Let us assume for simplicity $val(x_k,0)=0$ and $val(x_k,1)=1$.
For all vertices $x_t$ on the $k-th$ level such that $x_t \neq x_{k}$
we fix the value of the variable to $d_t$ such that $val(x_t,d_t)=0$.
This construction reduces the function $f$ to the $AND_k$ function,
implying $Q_E(f) \geq k$. As we know $D(f)=k$, this implies $Q_E(f)=k$.

}
\scase{ At most one of the children of $x_1$ query a variable that appears more than once in the decision tree.}
{
In this case there exists a $k$-depth path consisting of nodes querying $x_1,x_2, \ldots x_k$
such that each of these variables appear only once in the tree such that 
$$val(x_i,d_i)=x_{i+1}, 1 \leq i \leq k-1 \text{ and  } val(x_k,d_k)=1$$.

Now let the variable that is queried twice be $x_{dup}$ and
the nodes querying the variable be denoted as $x_{dup}^1$ and $x_{dup}^2$. 
If at most one of these nodes is in the $k$-th level then 
we can simply follow the method of the first case to reduce the function into 
$\displaystyle \prod_{i=1}^{k} (x_i \op \bar{d_i})$.

If both the node querying $x_{dup}$ are in the $k$-th level, then at least 
one of their parent nodes do not belong to the set $\{x_1,x_2 \ldots x_{k} \}$.
Let the variable being queried by that node be $x_{par}$ and it is parent of at-least $x_{dup}^1$.
We fix the value of $x_{par}$ to be $c$ such that $val(x_{par},\bar{c})=x_{dup}^1$.
Now we again fix all the value of the variables $x_t$ on the $k$-th level
except $x_{dup}^1$ and $x_{k}$ in the same way as in case 1 
to reduce the function to $\displaystyle \prod_{i=1}^{k} (x_i \op \bar{d_i})$.

The function $\displaystyle \prod_{i=1}^{k} (x_i \op \bar{d_i})$ is isomorphic to
the $AND_k$ function and thus the proof is completed.

}
\end{caseof}
\end{proof}

\section{Parity Decision Trees and Separation results}
\label{sec:3}
We now explore the parity decision tree model introduced in \cite{exact}.
This model is constructed using the fact that in the exact quantum query model,
the value of $x_{i_1} \op x_{i_2}$ can be evaluated using a single query.

A parity decision tree is similar to a deterministic decision tree.
But while in a decision tree a query can only return the value of a variable $x_i$,
in a parity decision tree a query can return either the value of a variable $x_i$ 
or the parity of two variables $x_{i_1} \oplus x_{i_2}$.
A parity decision tree represents a quantum algorithm in which the oracle is queried values of type
$x_{i_1}$ and $x_{i_1} \oplus x_{i_2}$. 
In fact in this case the work qubits can be measured after each query and reset to a default state.

Let $f$ be a Boolean function that can be expressed as a $k$-depth decision tree 
in which every internal node either queries a variable $x_i$ 
or the parity of two variables, $x_{i_1} \op x_{i_2}$.
We can then say that $Q_E(f) \leq k$.
Figure \ref{fig:paritytree1} gives an example of a parity decision tree.

\begin{figure}[!htb]
\begin{center}
\includegraphics[scale=0.27]{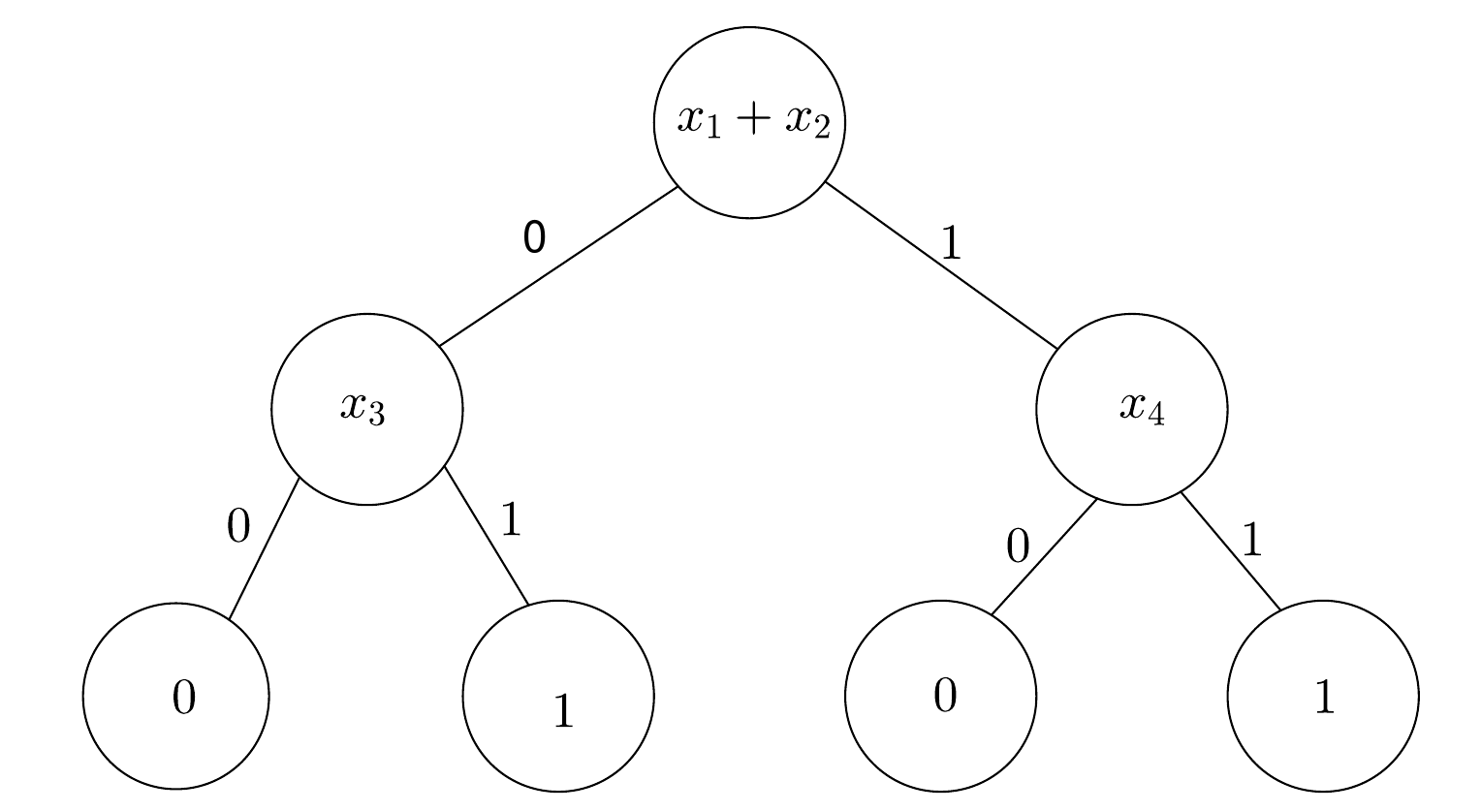}
\end{center}
\caption{Example of a parity decision tree}
\label{fig:paritytree1}
\end{figure}

The corresponding Boolean function is $(x_1 \op x_2)x_4 \op (x_1 \op x_2 \op 1)x_3$, with 
deterministic query complexity $3$ and exact quantum query complexity $2$.

A $k$-depth parity decision tree can only evaluate a function of algebraic degree less than or equal to $k$,
whereas there may exist a Boolean function of degree higher than $k$ that can be evaluated using 
$k$ queries. Thus, although this model does not completely capture the power of the quantum 
query model, we use the generalized structure of this model to find separable query friendly 
functions for certain values of $n$.

We say a parity decision tree $T$ completely describes a Boolean function $f$
if $T$ is a parity decision tree with the minimum depth (say $depth_f$)
among all parity decision trees that represent $f$ and
$Q_E(f)$ is equal to $depth_f$. 

\begin{lemma}
\label{lemx1}
Given any $k$ there exists a Boolean function $f$ with $2^{k+1}-2$ variables such that $Q_E(f)=k$.
\end{lemma}
\begin{proof}
This proof follows directly from the definition of parity decision trees and the proof of 
existence of a Boolean function with $2^k-1$ variables with $D(f)=k$.
Again we construct a $k$ depth complete parity decision tree such that every internal 
node is a query of the form $x_{i_1} \op x_{i_2}$ such that no variable appears twice 
in the tree. This tree represents a Boolean function $f$ of $2(2^k-1)$ variables and 
inherently $Q_E(f) \leq k$. 
This function can also be reduced to the $AND_k$ function which implies $Q_E(f) \geq f$.
This implies $Q_E(f)=k$.
We skip the proof of reduction to avoid repetition.

This is also the maximum number of influencing variables that a function $f$ can have
so that $Q_E(f)=k$ and $f$ can be completely described using parity decision trees.
This can be proven in the same way as in lemma \ref{lem:l0} and we do not repeat it
for brevity.
\end{proof}

We now prove some observations related to separability for a broader class of
functions and then explore separability in query friendly functions.
\begin{theorem}
If $n \neq 2^k-1$ for any $k$, then there exists a Boolean function for which $Q_E(f)< DQ_n$.
\end{theorem}
\begin{proof}
Let $2^{k-1}-1 < n <2^{k}-1$ for some natural number $k$.
In this case $DQ_n=k$.
However, there exist Boolean functions $f_Q$ with $n$ influencing variables such that $Q_E(f_Q)=k-1$.
We define a generic class of such functions using parity decision trees.
Let $n=2^{k-1}-1+y$.
Then we can always construct a complete parity decision tree of depth $k-1$ with the following constraints:
\begin{itemize}
\item Every variable appears only once in the tree.
\item $y$ internal nodes have query of the form $x_{i_1} \op x_{i_2}$. The rest of the internal nodes
query the value of a single variable.
\end{itemize}
Since $y \leq 2^{k-1}-1$, which is the number of internal nodes in a complete parity decision tree
of depth $k-1$, such a function always exists.  
\end{proof}

However, if $n=2^k-1$ for some $k$, then there not does not exist any Boolean function $f$
that can be completely expressed using the parity decision trees such that $Q_E(f)<DQ_n$.
If $n = 2^k-1$ then $DQ_n=k$ as well and
there does not exist any Boolean function $f_Q$ with $n$ variables 
that can be expressed using parity trees and has $Q_E(f_Q) \leq k-1$. 
This is true as we have already obtained that the Boolean function with maximum number of influencing variables
and depth $k-1$, that can be expressed using parity decision tree is $2^k-2$ 
(putting $k-1$ in place of $k$ in Lemma~\ref{lemx1} above).

Moreover, there does not exist any Boolean function
with $3$ influencing variable such that  
exact query complexity is less than $DQ_3$, which is equal to $2$.
It is interesting to note that
if for some $n=2^k-1$ there exists a Boolean function with $Q_E(f)=k-1$ then there exists separation for
all $n=2^j-1:j>k$. This can be easily proven with induction.
\begin{lemma}
If there exists a function $f_k$ with $2^k-1$ influencing variables such that $Q_E(f)=k-1$,
then there exists a function $f_j$ with $2^j-1$ influencing variables such that $Q_E(f) \leq j-1$ for all $j>k$.
\end{lemma}
\begin{proof}
If there exists a function $f_k$ with the specified property then $f_{k+1}$ can be constructed as follows.
$f_{k+1}=x_{2^{k+1}-1}(f_k(x_1,x_2,\ldots x_{2^k-1}) \op (x_{2^{k+1}-1} \op 1)f_k(x_{2^k},x_{2^k+1},\ldots x_{2^{k+1}-2})$.
It is easy to see $Q_E(f_{k+1}) \leq k$.
Using this construction recursively yields a desired function for any $j>k$.
\end{proof}

We complete the categorization by defining a generalized subclass of Query friendly Boolean functions.
We define this subclass such that a function $f$, belonging to this, has $Q_E(f) = DQ_n-1$.

\subsection{\label{subsec:2}Separable Query Friendly functions}
We construct a generic function for this set of query friendly functions 
using parity decision trees for values of $n$ such that there exists $k, 2^{k-1}-1 < n \leq 2^{k-1}+2^{k-2}-1$.
We first describe the construction using a parity decision tree and then prove the
query complexity values of the function. 

Let us construct a parity decision tree of depth $k-1$ in the following manner. 
The first $k-2$ levels are completely filled, with each internal node querying a single 
variable. All variable appears exactly once in this tree.
Let these variables be termed $x_1, x_2, \ldots, x_{2^{k-2}-1}$. 
In the $(k-1)$-th level, there are $\lceil \frac{n-(2^{k-2}-1)}{2} \rceil$ internal nodes,
with each query being of the form $x_{i_1} \op x_{i_2}$.
(In case $n-2^{k-2}+1$ is odd, there is one node querying a single variable).
Then if $n=2^{k-1}$ there are $2^{k-3}+1$ internal nodes in $(k-1)$-th level
and if n=$2^{k-1}+2^{k-2}-1$ there are $2^{k-2}$ nodes in the $(k-1)$-th level,
resulting in a fully-complete binary tree of depth $k-1$.
We denote this generic function as $f_{(n,2)}$.

\begin{theorem}
The Boolean function $f_{(n,2)}$ on $n$ influencing variables has $D(f)=DQ_n$ and $Q_E(f)=DQ_n-1.$
\end{theorem}
\begin{proof}
If $2^{k-1}-1<n \leq 2^{k-1}+2^{k-2}-1$ then $DQ_n=k$. 
We first prove that $Q_E(f_{(n,2)})=k-1$.
Since there exists a parity decision tree of depth $k-1$,
\begin{align}
Q_E(f_{(n,2)})\leq k-1.
\label{eq:qf21} 
\end{align}
If we fix one of the variables of each query of type $x_{i_1} \op x_{i_2}$ to zero then
the reduced tree corresponds to a non-separable function shown in \ref{subsec:1}
of depth $k-1$,that is the function can be reduced to $AND_{k-1}$.
This implies
\begin{align}
Q_E(f_{(n,2)}) \geq k-1.
\label{eq:qf22} 
\end{align}
Combining \eqref{eq:qf21} and \eqref{eq:qf22} we get $Q_E(f_{(n,2)})=k-1$.

Now we show that $D(f_{(n,2)}) = k$ by converting the parity decision tree to a 
deterministic decision tree of depth $k$.
All the internal nodes of the parity decision tree from level $1$ to level $k-2$ queries a single variable.
The nodes in the $k-1$-th level have queries of the form $x_{i_1} \op x_{i_2}$. 
Each such node can be replaced by a deterministic tree of of depth $2$ in the following way.
Suppose there is a internal node $x_{i_1} \op x_{i_2}$ in the $(k-1)$-th level.

We replace this node with a tree, whose root is $x_{i_1}$.
Both the children of the node queries $x_{i_2}$ and the leaf node values are swapped in 
the two subtrees. Without loss of generality, suppose in the original tree
$val(x_{i_1} \op x_{i_2},0)=0$ and $val(x_{i_1} \op x_{i_2},1)=1$
Then in the root node $val(val(x_1,0),0)=0$ and $val(val(x_1,1),0)=1$
and so on.
Figure \ref{fig:conv} gives a pictorial representation of the transformation.
The resultant deterministic decision tree is of depth $k$ as there is at least $2^{k-3}$ 
node in the $k-1$-th level in the parity decision tree which goes through transformation.
This implies $D(f_{(n,2)}) \leq k$.
We also know that in this case $DQ_n=k$. Combining the two results we get $D(f_{(n,2)})=k$.
\end{proof}

\begin{figure}[!htb]
\begin{center}
\includegraphics[scale=0.26]{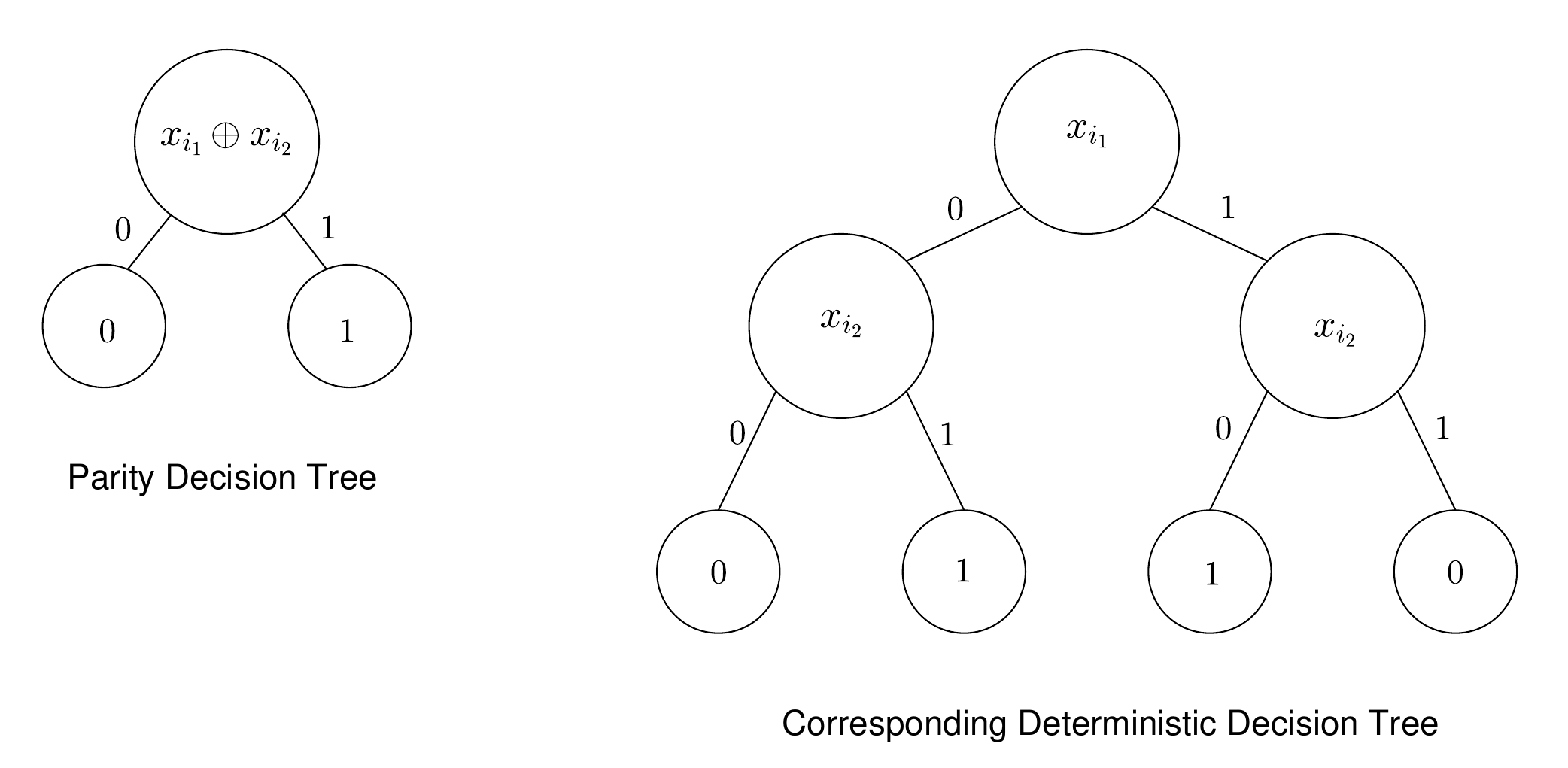}
\end{center}
\caption{Conversion of a node in the parity decision tree to a deterministic decision tree}
\label{fig:conv}
\end{figure}

\begin{remark}\label{remark1}
It should be noted that although we use a particular function $f$ for any $n$ to show the
separation for $Q_E(f)$ and $D(f)$, this immediately means that this separation is established 
for at least the class of functions on $n$ influencing variables that are PNP equivalent to $f$.
\end{remark}

Let us now consider a function of the form $f_{(5,2)}$ described by its ANF as
below:
\begin{align*}
f&=(x_1 \op 1)(x_2 \op x_3) \op x_1(x_4 \op x_5)\\
&=x_1x_2 \op x_1x_3 \op x_1x_4 \op x_1x_5 \op x_2 \op x_3.
\end{align*}
This provides an example for $n=5, D(f)=3$, and $Q_E(f)=2$.
In Figure \ref{fig:7} we present the decision tree for 
this function and the corresponding quantum circuit is provided
in Figure \ref{fig:8}.

\begin{figure}[!htb]
\begin{center}
\includegraphics[scale=0.275]{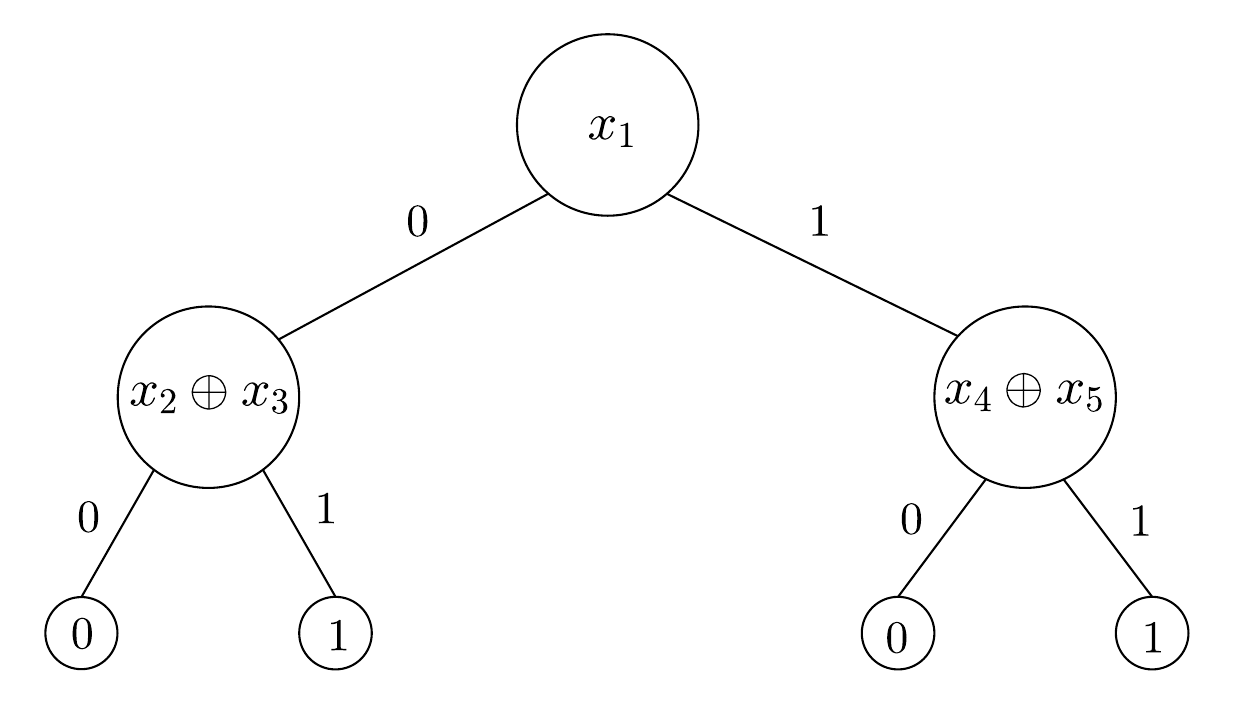}
\end{center}
\caption{Parity Decision Tree corresponding to $f_{(5,2)}$}
\label{fig:7}
\end{figure}

\begin{figure}[!htb]
\begin{center}
\includegraphics[scale=0.275]{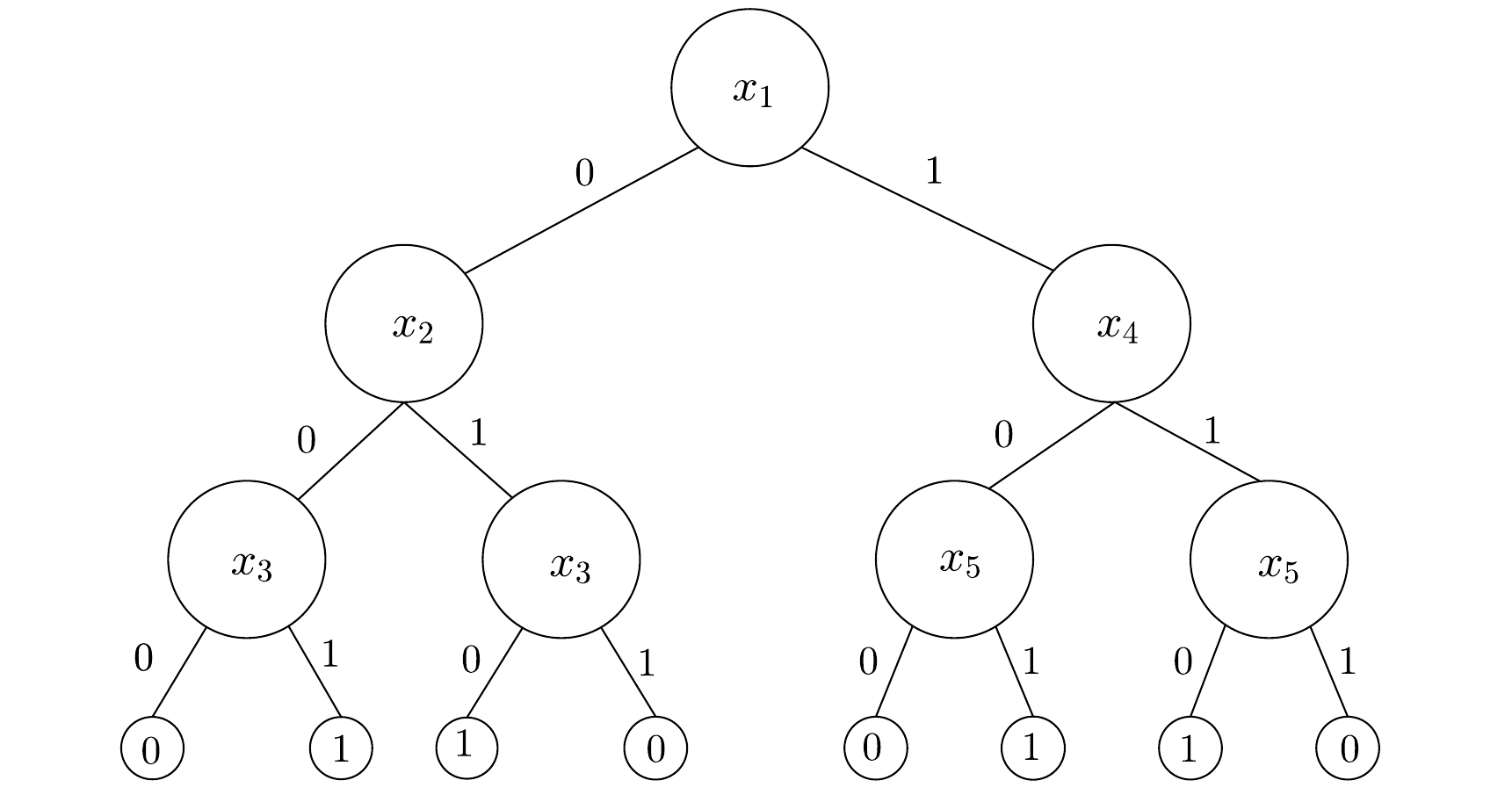}
\end{center}
\caption{Classical Decision Tree corresponding to $f_{(5,2)}$}
\label{fig:7}
\end{figure}

\begin{figure}[!htb]
\begin{center}
\includegraphics[scale=0.25]{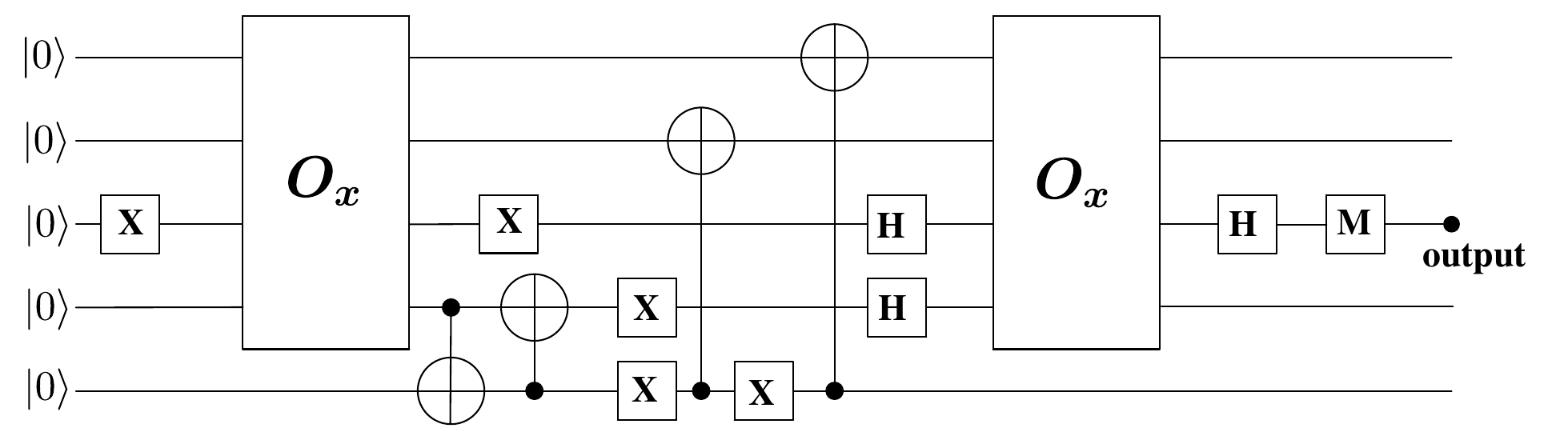}
\end{center}
\caption{Quantum algorithm responding to $f_{(5,2)}$}
\label{fig:8}
\end{figure} 

We now explain for the sake of completeness the difference in working 
of the exact quantum and deterministic algorithm for this function. 

Suppose we want to evaluate this function at the point $(1,0,1,0,1)$.
The deterministic algorithm will first query $x_1$,
and getting its value as $1$ it will then query $x_4$.
Since $x_4$ is $0$ it will query the $x_5$ node which is it's left children 
and then output $1$ as $x_5$ is $1$.

The quantum algorithm will evaluate as follows.
\begin{enumerate}
\itemsep0em 
\item Here $\psi_{start}=\ket{0}\ket{0}\ket{0}\ket{0}\ket{0}$.
\item The first {\sf X} gate transforms it into $\ket{1}_2\ket{0}\ket{0}$

{\em Here $\ket{i}_2$ implies $\ket{a}\ket{b}\ket{c}$ where $abc$ is the binary representation of integer $i$.}
\item Then we get $O_x (\ket{1}_2\ket{0}\ket{0})=\ket{1}_2\ket{x_1}\ket{0}=\ket{1}_2\ket{1}\ket{0}$.
\item The CNOT gates, the not gate and the Hadamard gates ({$\sf H^3$} and {$\sf H^4$}) transform
the state into $(\frac{\ket{4}_2+\ket{5}_2}{\sqrt{2}})\ket{-}\ket{0}$
where $\ket{-}=\frac{\ket{0}-\ket{1}}{\sqrt{2}}$.
\item Now 
\begin{align*}
&O_x(\frac{\ket{4}_2+\ket{5}_2}{\sqrt{2}})\ket{-}\ket{0}= \\
&(\frac{(-1)^{x_4}\ket{4}_2+(-1)^{x_5}\ket{5}_2}{\sqrt{2}})\ket{-}\ket{0}.
\end{align*}
Let this state be $\ket{\phi}$.
\item ${\sf H^3}\ket{\phi}=\frac{1}{2}((-1)^{x_4}+(-1)^{x_5})\ket{4}_2+((-1)^{x_4}-(-1)^{x_5})\ket{5}_2)\ket{-}\ket{0}$
\item since $x_4=0$ and $x_5=1$ we get $\ket{5}_2\ket{-}\ket{0}$
which is equal to $\ket{1}\ket{0}\ket{1}\ket{-}\ket{0}$.
Measuring the third qubit in computational basis we get the desired output, $1$.
\end{enumerate}
This completes the example of separation.

Finally, we conclude this section by proving that our construction
of separable query friendly function indeed
finds such examples for all cases where a parity decision
tree can compute such a function. This completes the 
characterization using parity decision trees.
\begin{theorem}
If $2^{k-1}+2^{k-2}-1 < n \leq 2^k-1$, there does not exist any
separable query friendly function that can be  
completely described using parity decision trees.
\end{theorem}
\begin{proof}
Let $f_n$ be a query friendly function on $2^{k-1}+2^{k-2}+1 < n \leq 2^k-1$ influencing variables.
In this case $DQ_n = k$, and hence $D(f_n)=k$. 
Therefore there exists a corresponding $k$-depth decision tree $T_f$.
As we know there are at most $2^k-1$ internal nodes in such a tree
and at least $2^{k-1}+2^{k-2}$ variables that needs to be 
queried at least once.
Therefore there can be at most $2^{k-2}-1$ internal nodes
which query variables that appear more than once in the tree.

This implies that there exists a node in the $k$-th level querying
a variable  $x_{i_{k}^0}$ such that it appears only once in the decision tree. 
We consider the root($x_{i_1}$) to $x_{i_{k}^0}$ path. It is to be noted 
that the root variable needs to be queried only once in any optimal tree.
Let us also assume for simplicity that $val(x_{{i_k}^0},0)=0$ and
\begin{align*}
&val(x_{i_t},d_t)=x_{i_{t+1}}, \ 1 \leq t \leq k-2 \\
&val(x_{i_{k-1}},d_{k-1})=x_{i_{k}^0}
\end{align*}

Let us now define the following sets of variables:

\begin{align*}
& W_j \subseteq \{ x_1,x_2,\ldots x_n \} \\
& X_j=W_j \cup \{ x_{i_{k}^0} \} \\
&Y_j \subseteq (\{ x_1,x_2,\ldots x_n \} \setminus \{ x_{i_{k}^0} \} ) \\
&\text{where}~ 1 \leq j \leq k
\end{align*}
Let $g_j \text{ and }h_j,1 \leq j \leq k$ be functions with influencing variables belonging from 
the sets $X_j,Y_j$ respectively.
Then the ANF of $f_n$ can be described as:
\begin{align*}
& f_n=(x_{i_1} \op \overline{d_1})g_1(X_1) \op (x_{i_1} \op d_1)h_1(Y_1)
\end{align*}
This is because the variable $x_{i_{k}^0}$ 
can influence the function if and only if $x_{i_1}=d_1$. This is due to the fact
that $x_{i_{k}^0}$ is queried only once in the decision tree. 
Similarly,
\begin{equation*}
g_1(X_1)=(x_{i_2} \op \overline{d_2})g_2(X_2) \op (x_{i_2} \op d_2)h_2(Y_2),
\end{equation*}
and so on. Finally we have 
\begin{equation*}
g_{k-2}(X_{k-2})=(x_{i_{k-1}} \op \overline{d_{k-1}})x_{i_{k}^0} \op (x_{i_{k-1}} \op d_{k-1})h_{k-2}(Y_{k-2}).
\end{equation*}

Therefore, the function $f_n$ can be written as 
\begin{align*}
&f_n= (x_{i_1} \op \overline{d_1})(x_{i_1} \op \overline{d_2})\ldots
(x_{i_{k-1}} \op \overline{d_{k-1}})x_{i_{k}^0} \op h_{k-1}(Y_k).
\end{align*}
This, in turn, implies that the resultant ANF
contains a $k$-term monomial $x_{i_1}x_{i_2}\ldots x_{i_{k}^0}$, which implies $\deg(f) \geq k$.

It has been shown in \cite[3.1]{exact} that
the minimum depth of any parity decision tree completely describing $f$
is at equal to or greater than $\deg(f)$, 
which implies there does not exist any query friendly function 
that can be completely described with a parity decision tree 
of depth $k-1$. This concludes our proof.
\end{proof}

With this proof of limitation we conclude the study of Query friendly functions in this paper. 
Next we study the deterministic and exact quantum query complexity of a large class of Boolean 
functions.

\section{Maiorana McFarland Bent Functions}\label{sec:4}
In this section we observe how parity decision trees can give us separation
in a large class of Cryptographically important Boolean functions. 
We consider the Maiorana-McFarland (M-M) type Boolean functions~\cite{MM1}, defined as follows.

\begin{definition}
Given any positive integer $n$ a Boolean function of M-M class
on $n=n_1+n_2$ variables $(v_1,v_2,\ldots v_n)$ is defined as 
$$f(x,y)=\phi(x).y \op h(x) ~, x \in \{0,1\}^{n_1},y \in \{0,1\}^{n_2}.$$
where
\begin{enumerate}

\item $x$ represents the variables $x_1=v_1,x_2=v_2,\ldots ,x_{n_1}=v_{n_1}$ 
and $y$ represents the variables ${y_1=v_{n_1}+1},y_2=v_{n_1}+2,\ldots ,
y_{n_2}=v_n$.

\item $h$ is any Boolean function and $\phi$ is any map $\phi: \{0,1\}^{n_1} \rightarrow \{0,1\}^{n_2}$.

\item $a.y$ is defined as the linear function $\displaystyle \bigoplus_{a_i=1} y_i$.
\end{enumerate}
\end{definition}

If we set $n_1=n_2=\frac{n}{2}$ and define $\phi$ to be a bijective mapping, 
all resultant M-M functions are bent functions~\cite{Bent1}, 
which are functions with highest possible nonlinearity for a given even $n$.
The non linearity of a function is defined as the minimum hamming distance of the 
truth table of a function of $n$ variable from all the linear function truth tables 
on $n$ variables cite~\cite{MM1}.
The M-M Bent functions and its different modifications have extensive applications in cryptographic
primitives and in coding theory ~\cite{Bentn2}.
 
We denote this class of M-M Bent functions by $\B_n$. 
There are $2^{2^{\frac{n}{2}}}(2^{\frac{n}{2}}!)$ functions in this class and the algebraic degree 
of the functions in this class vary between $2$ and $\frac{n}{2}$.
It is important to note that many functions of this class are not PNP equivalent, as two functions 
with different algebraic degree can not be PNP equivalent. 
At the same time it is also not necessary for two functions in $\B_n$ with same algebraic degree to be 
PNP equivalent. 
For an example, let us consider the functions corresponding to the identity permutation map, i.e. $\phi(i)=i$.
Then the function is of the form 
$ \big( \displaystyle \bigoplus_{i=1}^{\frac{n}{2}}x_iy_i \big) \op h(x)$.   
Now let there be two functions such that that the function defined on $x$ $(h(x))$
are not PNP equivalent. Then the two functions are not PNP equivalent as well. 

Having discussed the diversity of this class, we now analyze how the underlying definition of this class
can lead to the same bounds for all the functions belonging to this class, and we use parity decision tree
to achieve these bounds. 

\subsection{Deterministic and Exact Quantum Query Complexity}
We first calculate the deterministic query complexity of any function in the $\B_n$.
Given a point $\alpha \in \{0,1\}^{\frac{n}{2}}$ we define the point $\alpha^{(i)}, 1 \leq i \leq \frac{n}{2}$
as follows. 
\begin{align*}
& 1 \leq j \leq \frac{n}{2}, ~j \neq i ~:~  \alpha^{(i)}=\alpha_j \\ 
& j=i ~:~ \alpha^{(i)}_j = \overline{\alpha_j}
\end{align*}

We also define the points $A^1, A^0 \in {0,1}^{\frac{n}{2}}$ so that $A^1_i=1 ~\forall i$ and $A^0_i=1 ~\forall i$.
 
\begin{theorem}\label{th4:1a}
The deterministic query complexity of any function in $\B_n$ is $n$.
\end{theorem}
\begin{proof}
Let us assume that there exists a deterministic decision tree $D$ that 
queries $n-1$ variables to evaluate a function $f \in \B_n$ in the worst case. 
Let $D(x,y)$ denote the output obtained using the Deterministic tree with $(x,y)$ as
the input. 
This means the longest root to leaf vertex contains $n-1$ internal nodes (queries).

We consider the point $\hat{x} \in \{0,1\}^{\frac{n}{2}}$ such that $\phi(\hat{x})=A^1$.
Then $f(\hat{x},y)= y_1 \op y_2 \op \ldots y_{\frac{n}{2}} \op h(\hat{x})$ for all $y \in \{0,1\}^{\frac{n}{2}}$.
Therefore at any point $(\hat{x},y)$ any deterministic decision tree (algorithm) 
has to query all $\frac{n}{2}$ bits of $y$ to evaluate the function correctly.

Now if a decision tree doesn't query a variable  $x_i \in X$ at a point $(\hat{x},y)$
Then the decision tree traversal for the points $(\hat{x},y)$ and $(\hat{x}^{(i)},y)$
will be identical, so that $D(\hat{x},y)=D(\hat{x}^{(i)},y)~ \forall y \in \{0,1\}^{\frac{n}{2}}$.   

But weight of $(\phi{\hat{x}^{(i)}})$ is at most $n-1$.
This implies for any point $(\hat{x}^{(i)},y)$ there is at least an index $1 \leq k \leq n$
such that $f(\hat{x}^{(i)},y)=f(\hat{x}^{(i)},y^{(k)})$.

However we know from the definition of $\hat{x}$ that $f(\hat{x},y) \neq f(\hat{x},y^{(i)})$.
This contradicts the claim that a deterministic decision tree can evaluate a function $f \in B_n$
with $n-1$ queries in the worst case, and thus we have $D(f)=n$.
\end{proof}

Now we observe how parity decision tree can be used to form a quantum algorithm which can 
always evaluate a function in $\B_n$ with less than $n$ queries. 

We first provide a very simply derivable quantum advantage using parity decision trees.

\begin{lemma}\label{lem4:1}
Given any function $f \in \B_n$ we have $Q_E(f)\leq \lceil \frac{3n}{4} \rceil$.
\end{lemma}
\begin{proof}
We prove this by describing an algorithm that can evaluate any Boolean function in of the type 
$\B_n$ using $\lceil \frac{3n}{4} \rceil$ queries.

The queries made by this quantum algorithm are of the form $x_i$ or $x_{i_1} \op x_{i_2}$ 
and can therefore be expressed as a parity decision tree.

Given any input $(x,y)$ the algorithm first queries the $\frac{n}{2}$ variables $x_1,x_2,\ldots, x_n$.
Then depending on the definition of the function it does one of the following two tasks. 
\begin{enumerate}
\item If $h(x)=0$ then it evaluates $\displaystyle \bigoplus_{\phi(x)_i=1} y_i$

\item If $h(x)=1$ then it evaluates $ \Big( \displaystyle \bigoplus_{\phi(x)_i=1} y_i \Big) \oplus 1$
\end{enumerate}
In either case this requires $\lceil \frac{wt(\phi(x))}{2} \rceil$ queries and therefore at max requires
$\lceil \frac{n}{4} \rceil$ queries. This proves the upper bound.
\end{proof}

\begin{figure}[!htb]
\begin{center}
\includegraphics[scale=0.29]{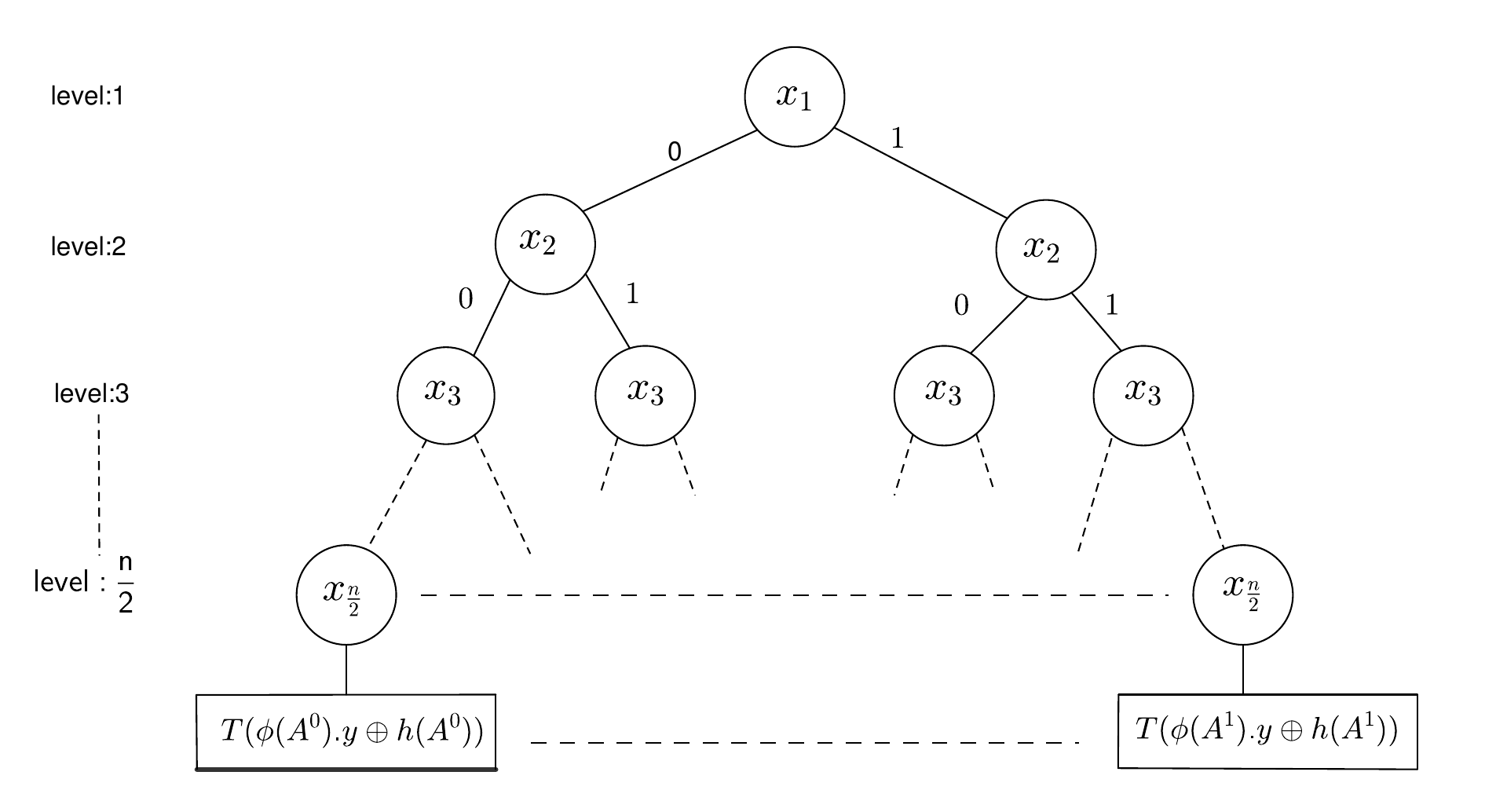}
\end{center}
\caption{(Parity) Decision tree structure for the functions in $\B_n$ }
\label{fig:9}
\end{figure} 

It is also also evident from the constructions that the parity decision tree corresponding to the
quantum algorithm and the classical decision tree corresponding to the classical algorithm are 
analogous in nature.  Figure \ref{fig:9} shows the structure of the tree.
Here $T(\phi(x).y+h(x))$ denotes the parity decision tree (decision tree) for the quantum (classical) algorithm.
In case of the parity decision tree, $T(\phi(a).y+h(a))$ has a depth of
$\lceil \frac{wt(\phi(a))}{2} \rceil$ where as in the classical case it is $wt(\phi(a))$.
We do not draw the structures for these cases to avoid repetition.

However one should note this does not prove that there cannot be a parity based decision tree 
that may evaluate such a function using lesser number of queries.
Having established the lower bound for the parity decision tree method we next observe some more 
general lower bounds using simple reduction. 

\begin{lemma}
For any function $f$ in $\B_n$, we have $Q_E(f) \geq \max(Q_E(h),\frac{n}{4})$.
\end{lemma}
\begin{proof}

For any function $f(x,y) \in \B_n$ if we fix the value of $x$ such that $\phi(x)=A^1$ then 
the function is reduced to $f'(y)=y_1\op y_2 \op \ldots y_{\frac{n}{2}}$ and therefore $Q_E(f) \geq \frac{n}{4}$.

Similarly, if we fix $y=A^0$ then the $f$ is reduced to $f''(x)=h(x)$ and thus $Q_E(f) \geq Q_E(h)$.

Combining these two values we have $Q_E(f) \geq \max (Q_E(h),\frac{n}{4})$.
\end{proof}

Finally we show the generic real polynomial that represents any function in the class $\B_n$.
It is known that any Boolean function $f$ can be represented by a unique multivariate polynomial
$p: \R^n \rightarrow \R$. The exact quantum and deterministic query complexity can be related to
the degree of this polynomial $(\deg^{\R}(p))$ as $Q_E(f) \geq \frac{\deg^{\R}(p)}{2}$ and $D(f) \geq \deg^{\R}(p)$.
 
This polynomial can in fact be derived from the description of the parity decision tree.
The degree of the polynomial corresponding to any function in $\B_n$ is found to be $n$ 
which gives a tighter lower bound of $\frac{n}{2}$ on the exact quantum query complexity.

\begin{lemma} \label{lem4:2}
The degree of the real polynomial corresponding to any function in $\B_n$ is $n$.
\end{lemma}

\begin{proof}
The polynomial corresponding to the function $f\in \B_n$ such that $f(x,y)=\phi(x).y \op h(x)$
can be formulated as follows. 

We observe that only one of the linear function defined on the variables $\{y_1,y_2,\ldots, y_{\frac{n}{2}} \}$
is evaluated for any input $(x,y)$ depending on the value of $\phi(x)$.
Therefore we first form the following product terms on the variables $\{x_1,x_2,\ldots, x_{\frac{n}{2}} \}$.
We define the $\mathcal{P}_a, a \in\{0,1\}^{\frac{n}{2}}$ as
$$ \mathcal{P}_a= \displaystyle \prod_{a_i=0} (1-x_i)\displaystyle \prod_{a_i=1} x_i.$$

$\mathcal{P}_a$ evaluates to $1$ iff $x=a$, $0$ otherwise. 
Now we append the corresponding linear functions defined by $\phi(a)$ to each of these 
product terms. We also account for the function $h(x)$ which evaluates to $h(a)$ for any input $(a,y)$.

Therefore the linear function to be evaluated is  $ \Big( \displaystyle \bigoplus_{\phi(a)_i=1} y_i \Big) \op h(a)$,
which is represented as 
$$ \mathcal{L}_a=  \frac{1-\Big((-1)^{h(a)}\displaystyle\prod_{\phi(a)_i=1} (1-2y_i) \Big)}{2} $$

Therefore we have the polynomial $p(x)$ corresponding to the function $f$ as 
$$p(x)= \displaystyle \sum_{a \in \{0,1\}^n} \mathcal{P}_a \mathcal{L}_a.$$
Therefore by definition we have $\deg^{\R}(\mathcal{L}_a)=wt(a)$ and $\deg^{\R}(\mathcal{P}_a)=\frac{n}{2}, \forall a$ 
and since there is only one value of $a$ with $wt(a)=\frac{n}{2}$ this implies $\deg^{R}(p)=n$.

This polynomial is defined as $p:~\mathcal{R}^n \rightarrow R$ but its range becomes $\{0,1\}$ 
when the domain is restricted to $\{0,1\}^n$.
\end{proof}

This proof is also another way of showing that the Deterministic Query complexity of any function in $\B_n$ is $n$.

Combining Lemma~\ref{lem4:1} and Lemma~\ref{lem4:2} we obtain the following result.

\begin{theorem}\label{th4:1}
For any Maiorana McFarland type Bent function $f$ we have $\frac{n}{2} \leq Q_E(f) \leq \frac{3n}{4}$.
\end{theorem}

The statement of Theorem~\ref{th4:1} gives rise to the following corollary.

\begin{corollary}
For all values of $n$ there are two or more M-M Bent functions that have different algebraic degree 
and same exact quantum query complexity.
\end{corollary}
\begin{proof}
The algebraic degree of the functions in $\B_n$ vary between $2$ and $\frac{n}{2}$ where as the exact quantum 
query complexity varies between $\frac{n}{2}$ and $\lceil \frac{3n}{4} \rceil$. 

Therefore applying pigeonhole principle it is easy to see that there are at least two Boolean functions
with different algebraic degree and same exact quantum query complexity.
\end{proof}
It is important to note that two functions in $\B_n$ may have the same algebraic degree yet different 
exact quantum query complexity.
Characterizing these equivalence classes for $\B_n$ appears to be a very interesting problem.
It is also interesting to observe that the real polynomial corresponding to any function in $\B_n$ can be
obtained from the description of the corresponding parity decision tree.

We further observe the exact quantum query complexity of functions in $\B_4$ which gives us more insight 
into this problem.
The different Boolean functions in the class $\B_4$ upto isomorphism are the following. 
\begin{align*}
& f_1(x)=x_1x_2 \op x_3x_4 \\
& f_2(x)=x_1x_2 \op x_3x_4 \op x_2x_3
\end{align*}

The exact quantum query complexity of all $4$ variable M-M Bent functions can be 
observed from \cite[Table A.1]{exact} which is obtained using the convex optimization package CVX~\cite{cvx} for Matlab. 
In this regard we observe that $Q_E(f_1)=Q_E(f_2)=3$ which touches the upper bound of 
$\lceil \frac{3n}{4} \rceil$ in this case. 

\section{conclusion}\label{sec:5}
In this paper we have first discussed the separation between the deterministic and exact quantum query model in terms 
of the number of influencing variables in Section~\ref{sec:2} and \ref{sec:3}. 
We have used the parity decision tree model to find separation between deterministic
and exact quantum query in a special class of Boolean functions (Query Friendly functions) using the structured 
nature of the parity decision tree model.
The characterization achieved by us in terms of 
query friendly functions is as follows. 

\begin{enumerate}

\item For all varies of $n$ there exists a non-separable query friendly function.

\item If $n=2^k-1$ or $n=2^k-2,~k>2$ , then all query friendly functions are non-separable.

\item If $n \neq 2^k-1$, then we construct a set of non-separable functions, namely $f_{(n,1)}$. 

\item If $2^{k-1}-1< n \leq 2^{k-1}+2^{k-2}-1$, then we construct a set of separable functions, 
namely $f_{(n,2)}$.

\item If $2^{k-1}+2^{k-2}-1 < n \leq 2^k-1$, we show that no separable function on $n$ variables
can be completely described using parity decision trees.
\end{enumerate}

In this regard we have observed the following open problems
which shall exhaustively determine the limitation of the parity decision tree model in these cases. 
The problems are as follows:
\begin{enumerate}
\item Does there exist a function $f_1$ with $n=2^k-1$ influencing variables such that
$Q_E(f_1) < k$?
\item Does there exist a separable query friendly function $f_2$ with $n$ influencing variables, 
where $2^{k-1} + 2^{k-2} - 1 < n < 2^k-2,~ k>2$?
\end{enumerate}
If any of the above problems yield a negative result that would imply the parity decision 
tree model indeed completely characterizes the functions in such a scenario.

Then we have analyzed the deterministic and exact quantum query complexity of the
class of M-M Bent functions ($\B_n$) in Section~\ref{sec:4}. 
We have used the parity decision tree method to obtain advantage and provide a simple 
generalized query algorithm for the $\B_n$ class of functions.
The results in this direction are as follows.

\begin{enumerate}

\item The deterministic query complexity of any function in $\B_n$ is $n$.

\item The exact quantum query complexity of any function in $\B_n$ is at most $\lceil \frac{3n}{4} \rceil$.

\item The Polynomial degree of any function in $\B_n$ is $n$, which implies $ Q_E(f)>=\frac{n}{2}~\forall~f \in B_n$. 

\end{enumerate}

The bounds obtained indicate that there are multiple M-M Bent functions on $n$ variable for any $n$ such that
they have the same exact quantum query complexity but are not PNP equivalent and have different algebraic degree.
Characterizing the different exact quantum query equivalence classes seem to be a very interesting problem, 
one that may further concretize the relation between the structure of a Boolean function and its query complexity
in different models.
The primary open problem in this direction is finding out the exact query complexity of all the functions
in $\B_n$ and thereby generalizing which functions in $\B_n$ are not PNP equivalent yet have the same query complexity.
Studying the different modified classes of M-M Bent functions are also of importance, and we hope the results obtained in 
this paper will be helpful towards forming a generalized quantum query algorithm for the Maiorana McFarland class of functions.

\end{document}